%% file: arxiv_v1.tex
\documentclass[journal]{IEEEtran}

\usepackage{latexsym}
\usepackage{amssymb}
\usepackage{bm}
\usepackage[dvips]{graphics}
\usepackage{graphicx}
\usepackage{subfigure}
\usepackage{setspace}
\usepackage{psfrag}
\usepackage{amsfonts,amsmath,amssymb,hyperref,amsthm}
\usepackage{algorithm}
\usepackage{algorithmic}
\usepackage{dsfont,color,subfigure,setspace,psfrag,cite}

\input{defns.tex}

\newtheorem{remark}{Remark}

\newtheorem{theorem}{Theorem}
\newtheorem{lemma}{Lemma}

\begin{document}

\title{Outage Performance of Uplink  Two-tier Networks Under Backhaul Constraints}

\author{Shirin Jalali,  Zolfa Zeinalpour-Yazdi and H. Vincent Poor
      \thanks{S. Jalali is with the Department of   Electrical  Engineering, Princeton university, NJ 08540 (e-mail: sjalali@princeton.edu),}
      \thanks{Z. Zeinalpour-Yazdi is with the Department of Electrical and computer Engineering, Yazd University, Yazd, Iran (e-mail: zeinalpour@yazd.ac.ir),}
       \thanks{H. V. Poor is with the Department of   Electrical  Engineering, Princeton university, NJ 08540 (e-mail: poor@princeton.edu).}}

 \maketitle
 \newcommand{\p}{\mathds{P}}
\newcommand{\Lc}{\mathcal{L}}
\newcommand{\mb}{\mathbf{m}}
\newcommand{\bb}{\mathbf{b}}
\newcommand{\Xb}{\mathbf{X}}
\newcommand{\Yb}{\mathbf{Y}}
\newcommand{\Ub}{\mathbf{U}}
\newcommand{\La}{\Lambda}
\newcommand{\su}{\underline{s}}
\newcommand{\xu}{\underline{x}}
\newcommand{\yu}{\underline{y}}
\newcommand{\Xu}{\underline{X}}
\newcommand{\Yu}{\underline{Y}}
\newcommand{\Uu}{\underline{U}}
\newcommand{\ex}{{\rm e}}
\newcommand{\SIR}{{\rm SIR}}

\vspace{-4em}
\begin{abstract}
Multi-tier cellular communication networks constitute a promising approach  to expand the coverage of cellular networks and enable them to offer  higher data rates. In this paper, an uplink two-tier communication network is studied, in which macro users, femto users  and femto access points are geometrically  located inside the coverage area of a macro base station according to Poisson point processes. Each femtocell is assumed to have a fixed backhaul constraint that puts a limit on the maximum number of femto and macro users it can service. Under this backhaul constraint, the network adopts a special open access policy, in which each macro user is either assigned to its closest femto access point or to the macro base station, depending on the ratio between its distances from those two. Under this model,  upper and lower bounds on the outage probabilities experienced by users serviced by femto access points are derived as functions of the distance between the macro base station and the femto access point serving them. Similarly, upper and lower bounds on the outage probabilities of the users serviced by the macro base station are obtained.  The bounds in both cases are confirmed via simulation results.

\end{abstract}


\section{Introduction}
Fourth generation (4G) mobile communication standards such as LTE-advanced promise very high data rates. Enabling multi-tier  networks is one of the methods that enables such standards to address the ever-increasing  demand for higher data rates in cellular communication networks. In a multi-tier network, unlike the traditional design, multiple layers of cells, each  serviced by a different type of base station, are employed simultaneously. In two-tier femtocell networks, for example, in addition to the traditional base stations, there are femto access points (FAPs)  installed by users in their homes or offices. These additional base stations are connected to the cellular network through the users' broadband Internet connections. These  FAPs expand the coverage of the main network to indoors and also reduce its load. However, the limited capacities of users broadband connections impose a backhaul constraint that limits  the number of simultaneous users each femto cell can cover.

In this paper we study the outage performance of a two-tier uplink femtocell network. Macro users (MUs),  femto users (FUs) and FAPs are assumed to be spatially distributed according to  Poisson point processes (PPPs) \cite{HaenggiA:09}.  Each femtocell is assumed to have a limited backhaul capacity. Up to its capacity, each FAP  employs a special \emph{open access policy}, studied in \cite{CheungQ:12} and \cite{JoS:12} for downlinks. Based on this policy, each MU is serviced by its closest FAP if i) the ratio between its distance to  its closet FAP and its distance to the MBS exceeds some threshold, and ii) the number of users already being serviced by that FAP is less than its capacity.
\vspace{-.5em}
\subsection{Related work}

PPPs were originally suggested in \cite{BaccelliZ:97,BaccelliK:97,Brown:00} as a more tractable and realistic model for the locations of cells and users in a wireless network. The outage performance of two-tier networks  under PPP distribution of users or access  points is studied in \cite{DhillonG:12, Mukherjee:12,WangQ:12,YuM:12} and in \cite{ChandrasekharA:09, ChakchoukH:12,YuM:12,BaoL:13,BaoL:13-acm,ZeinalpourJ:14-tcom,ElSawyH:14} for downlink and uplink communications, respectively. In none of these papers are the FAPs' backhaul constraints taken into account. In fact,  to our knowledge, while there have been studies of the effects of  femtocell backhaul constraints on other aspects of networks, there has been no prior analytical work on their effects on the users' outage performance in a two-tier network. (Refer to \cite{XiaC:10,ElkourdiS:11, NgL:12,LoumiotisA:14} as a sample of some recent results.) In this paper, we extend the analysis of uplink tow-tier networks presented in \cite{ZeinalpourJ:14-tcom} to the case in which each FAP  has a backhaul constraint that limits the number of users it can service.  We derive  analytical upper and lower bounds on the outage probabilities experienced by the users serviced by the FAPs.
\vspace{-.8em}
\subsection{Notation}

Sets are denoted by calligraphic letters such as $\Ac$ and $\Bc$. The size of a set $\Ac$ is denoted by $|\Ac|$. The Laplace transform of random variable $X$ is denoted by $\Phi_{X}(s)\triangleq \E[\ex^{-sX}]$. Given $x\in\mathds{R}$, $(x)^+\triangleq\max(x,0)$.  Throughout the paper, $\Pc(s,x)$ denotes the  cumulative distribution function of a gamma random variable with shape parameter $s$ and scale parameter $1$. Given a Poisson random variable $X$ with parameter $\lambda$,  $\P(X\leq k)=\ex^{-\lambda}\sum_{i=0}^k{\lambda^k\over k!}=1-\Pc(k+1,\lambda).$
\vspace{-.8em}
\subsection{Paper organization}
The paper is organized as follows. Section \ref{sec:model} reviews the system model including the employed modulation, users and FAPs spatial distributions, and the access policy. The distributions of number of  users falling into different service groups  are studied in Section \ref{sec:dist}. Section \ref{sec:MU-by-FAP} studies  the outage probability experienced by the MUs serviced by FAPs. Similarly, Section \ref{sec:MU-by-MBS} analyzes the outage probability experienced by the MUs serviced by the MBS.  Section \ref{sec:num} presents  numerical results and, finally, Section \ref{sec:con} concludes the paper.
\vspace{-.8em}
\section{System model}\label{sec:model}
\subsection{MCFH technique}
Both macro and femto users are assumed to employ multicarrier frequency-hopping  (MCFH) modulation introduced in \cite{LanceK:97}. In MCFH the available bandwidth is divided into $n_s$ non-overlapping  subbands and each subband  is divided into  $n_h$  equispaced frequencies, respectively. Hence,  there are overall  $n_sn_h$ available orthogonal \emph{subchannels}.  During each time slot, each user   selects  $n_s$ subchannels by independently and uniformly at random choosing one subchannel from each subband. While MCFH modulation is very similar to  orthogonal frequency devision modulation (OFDM), unlike OFDM it does not require  centralized frequency assignment. Hence, while, with some minor adjustments,  the results derived under this modulation are also applicable to networks employing OFDM, MCFH modulation is much better suited for analytical performance studies.

%
\subsection{Spatial distribution}

Consider MBS $b_m$ located at the center of a circle of radius $R$ denoted by $\Sc_m$.  $\Ac_f$, $\Uc_m$ and $\Uc_f$ denote the set of FAPs, MUs and FUs, respectively. Conditioned on the locations of the FAPs $\Ac_f$, FUs and MUs are distributed according to independent PPPs.  FAPs and MUs are drawn according to PPPs of densities  $\lambda_f$ and $\mu_m$, respectively. Let   $\Uc_m$, $N_m=|\Uc_m|$, and $\bar{n}_{\rm mu}=\E[N_m]=\pi R^2 \mu_m$ denote the set of MUs in $\Sc_m$, the number of MUs and the expected number of MUs, respectively.
Similarly, let $\Ac_f$, $N_{a_f}=|\Ac_f|$, and $\bar{n}_{\rm fap}=\E[N_{a_f}]=\pi R^2 \lambda_f$ denote the set of FAPs in $\Sc_m$, the number of FAPs and the expected number of FAPs, respectively.
The FUs corresponding to each FAP $a_f\in\Ac_f$ are distributed according to a PPP with density $\mu_f$ in a disk of width $\delta$ and inner radius of $r_f$ centered at $a_f$. By this construction, the  expected number of FUs served by a femto cell is  equal to  $\bar{n}_{\rm fu}=\pi ((r_f+\delta)^2-r_f^2)\mu_f$.

Given FAP $a_f\in\Ac_f$, $\Uc_f(a_f)$ and $\Uc_m(a_f)$ denote the set of FUs and MUs, respectively, that are serviced by $a_f$.
Also  $N_f^{a_f}\triangleq |\Uc_f(a_f)|$ and $N_m^{a_f}\triangleq |\Uc_m(a_f)|$.  Finally, $\Uc_m(b_m)$ denotes the set of MUs serviced by the MBS $b_m$. Clearly, $\Uc_m=\cup_{a\in\Ac_f\cup\{b_m\}}\Uc_m(a).$
The number of MUs  covered by the MBS $b_m$ is denoted by $N_m^{b_m}$, \ie $N_m^{b_m}\triangleq |\Uc_m(b_m)|$. Note that, by definition,  $N_m=N_m^{b_m}+\sum_{a_f\in\Ac_f}N_m^{a_f}$.

\subsection{Access policy and backhaul constraint}\label{section:access}

We consider the open access scenario with  access parameter $\kappa\in[0,1]$, studied in \cite{CheungQ:12} for  downlink communications and in \cite{ZeinalpourJ:14-tcom} for uplink transmission, when the FAPs have no backhaul constraints. Let $d(u_m,a)$ denote the Euclidean distance between the (femto or macro) access point $a$ and $u_m$. Then, in this access model an MU is served by its nearest FAP $a_f$ if ${d(u_m,a_f) \over d(u_m,b_m)}$ is less than $\kappa$ and the backhaul constraint is not violated; otherwise it is served by the MBS.

To model the backhaul constraints, we assume that each FAP has access to a  fixed broadband capacity, which translates into covering at most $n_c$ users. The priority is always given to FUs. Once all FUs are serviced, if there is some remaining unused  capacity, it can be  allocated to  MUs. MU $u_m$ is potentially assigned to FAP $a_f$, if $d(u_m,a_f)\leq \kappa d(u_m,b_m)$. If there are more than one FAPs satisfying this condition, $u_m$  considers only  the closest one.  From all potential MUs  of an FAP $a_f$ with $N_f^{a_f}$ FUs, $a_f$  randomly chooses  up to $n_c-N_f^{a_f}$ of them to serve.
It is reasonable to assume that $n_c\geq \bar{n}_{\rm fu}$, or in other words,   the capacity of each FAP is at least as large as  the expected number of FUs in that cell.

In this model, due to the backhaul constraint, an MU can get arbitrarily close to an FAP $a_f$, and yet be serviced by the MBS. To avoid the arbitrarily large interference caused by such cases, we assume that, for any MU $u_m$, the ratio between its distances from any FAP $a_f$ and the MBS, \ie $d(u_m,a_f)/d(u_m,b_m)$, cannot be smaller than some threshold $\kappa_o$, where $\kappa_o\ll \kappa$. As argued in \cite{ZeinalpourJ:14-tcom}, this means that for an FAP $a_f$ located at distance $d$ from $b_m$, there exists a circle of radius ${\kappa_o\over 1-\kappa_o^2}d$ that includes $a_f$, where no MUs are allowed. In general, we can assume that $k_o$ depends on $d$, and  as a special case tune it such that the excluded circle of all FAPs have the same radius. While our analysis can be generalized to this case in a straightforward manner, to simplify the statement of the results, we assume that $k_o$ is fixed for all FAPs.

%

\vspace{-.5em}
\subsection{Channel Model}

To model the channel between user $u$ and access point $a$, $a \in \{b_m,a_f\}$, both small scale fading and path loss are considered. So it is assumed that the fading coefficients corresponding to the channel in subband $i\in [1:n_s]$ from user $u$ to  $a$, $H^i_{u,a}$, follows the Rayleigh distribution with parameter $\sigma^2$. Furthermore, we assume that the coefficients corresponding to different subbands and also different channels are all independent. The path loss is modeled as $\texttt{PL}_{u,a}=L_0 d_{u,a}^{\alpha},$
where $L_0$ is the path loss at unit distance, and $\alpha>2$ denotes the attenuation factor \cite{JoS:12}.

In this paper,  we assume that every user  employs power control to compensate for the effect of path loss. By power control, MUs serviced by the MBS intend to achieve a
received power level of $p_m$, and FUs and MUs serviced by  FAPs adjust their transmitted powers to achieve a received power  of  $p_f$.

\section{Users density distribution}\label{sec:dist}

\begin{figure}[b]
\vspace{-1.5em}
\begin{center}
\psfrag{bm}{ $b_m$}
\psfrag{af}{ $a_f$}
\psfrag{d}{ $r$}
\psfrag{r}{ $\hspace{-0.14cm}r_c$}
\includegraphics[width=5.5cm]{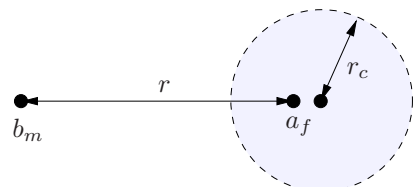}\caption{MUs served by  $a_f$ located at distance $r$ from  $b_m$.}\label{fig:MU-covered-FAP}\end{center}
\end{figure}

In this section, we study the distributions of the random variables $N_f^{a_f}$, $N_m^{a_f}$,  $N_m^{b_m}$ and $N_m$. As argued in \cite{ZeinalpourJ:14-tcom}, given FAP $a_f$ at distance $r$ from $b_m$, the set of points satisfying $d(u_m,a_f)\leq \kappa d(u_m,b_m)$ is the set of points inside a circle of radius  $r_c=({\kappa\over 1-\kappa^2})r$. (Refer to   Fig.~\ref{fig:MU-covered-FAP}.) The distance between the center of this circle and $b_m$ is equal to ${r\over 1-\kappa^2}$. For $\kappa\in(0,1)$, ${\kappa\over 1-\kappa^2}$ is an increasing function of $\kappa$, which implies that increasing $\kappa$ translates into increasing the coverage area of an FAP. As a special case, when $\kappa=0$, the FAP only covers FUs, and hence has a closed access policy.

For MUs serviced by FAPs, the potential coverage area of FAP $a_f$ located at distance $d$ from $b_m$  is a circle of radius $({\kappa\over 1-\kappa^2})d$. Let $\Uc_m(a_f)$ denote the  MUs that fall in the coverage area of FAP $a_f$.  Due to the backhaul constraint,  not all the MUs falling in $\Uc_m(a_f)$ can be serviced by $a_f$. Therefore, they can be partitioned into two groups, $\Uc_m^{\rm s}(a_f)$ and $\Uc_m^{\rm ns}(a_f)$, representing the   MUs that are serviced by $a_f$ and the MUs that fall in the coverage area of   $a_f$, but are serviced by $b_m$, respectively. Let   $N_{m,s}^{a_f}\triangleq |\Uc_m^{\rm s}(a_f)|$ and $N_{m,ns}^{a_f}\triangleq |\Uc_m^{\rm ns}(a_f)|$. Stochastically,
\[
N_{m,s}^{a_f}=\min(N_1,(n_c-N_2)^+),
\]
and
\[
N_{m,ns}^{a_f}=N_1-N_{m,s}^{a_f}=(N_1-(n_c-N_2)^+)^+,
\]
where $N_1$ and $N_2$ are independent and distributed as $ {\rm Poiss}(\bar{n}_{\rm mu}^{d})$ with
\begin{align}
\bar{n}_{\rm mu}^{d}\triangleq \pi \mu_m( ({\kappa\over 1-\kappa^2})^2-({\kappa_o\over 1-\kappa_o^2})^2)d^2\label{eq:n-bar-mu-d}
\end{align}
and $ {\rm Poiss}(\bar{n}_{\rm fu})$, respectively.
\begin{lemma}\label{lemma:1}
The Laplace transform of $N_{m,ns}^{a_f}$, the number of MUs that  fall in the coverage area of  FAP $a_f$ located at distance $d$ from the MBS $b_m$ but  serviced by $b_m$, satisfies
\begin{align*}
\Phi_{N_{m,ns}^{a_f}}(s|d_f)\;\leq\; &1-\Pc(n_c,\bar{n}_{\rm fu})+\ex^{\bar{n}_{\rm mu}^d (\ex^{-s}-1)}\Pc(n_c,\bar{n}_{\rm fu}).
\end{align*}
\end{lemma}
\begin{proof}
By definition, the number of MUs  in  $\Uc_m^{\rm ns}(a_f)$ can be written as $|\Uc_m^{\rm ns}(a_f)|=N_1-|\Uc_m^{\rm s}(a_f)|=(N_1-(n_c-N_2)^+)^+$,
where $N_1$ and $N_2$ are independent and distributed as $ {\rm Poiss}(\bar{n}_{\rm mu}^{d})$  and  $ {\rm Poiss}(\bar{n}_{\rm fu})$, respectively.
Therefore,
\begin{align}
\Phi_{|\Uc_m^{\rm ns}(a_f)|}(s|d_f)= &\E\left[\ex^{-s |\Uc_m^{\rm ns}(a_f)|}|{N_2 < n_c}\right]\P(N_2< n_c)\nonumber\\
&+\E\left[\ex^{-s |\Uc_m^{\rm ns}(a_f)|}|{N_2\geq n_c}\right]\P(N_2 \geq n_c)\nonumber\\
=&\E\left[\ex^{-s |\Uc_m^{\rm ns}(a_f)|}|{N_2 < n_c}\right]\P(N_2< n_c)\nonumber\\
&+\E\left[\ex^{-s N_1}\right]\P(N_2 \geq n_c),\label{eq:Nm-nc-Phi-step1}
\end{align}
where the last line follows from the independence of $N_1$ and $N_2$. Since $\E[\ex^{-s |\Uc_m^{\rm nc}(a_f)|}|{N_2 < n_c}]\leq 1$, from \eqref{eq:Nm-nc-Phi-step1},
\begin{align*}
\Phi_{|\Uc_m^{\rm nc}(a_f)|}(s|d_f)&\leq \P(N_2< n_c)+\E\left[\ex^{-s N_1}\right]\P(N_2 \geq n_c)\nonumber\\
&=1-\Pc(n_c,\bar{n}_{\rm fu})+\ex^{\bar{n}_{\rm mu}^d (\ex^{-s}-1)}\Pc(n_c,\bar{n}_{\rm fu}).
\end{align*}
\end{proof}

\begin{lemma}\label{lemma:2}
Let  $\gamma=({\kappa\over 1-\kappa^2})^2-({\kappa_o\over 1-\kappa_o^2})^2$,  and $
\beta \triangleq \Pc(n_c,\bar{n}_{\rm fu} )+ (1-\Pc(n_c,\bar{n}_{\rm fu}) )({\ex^{ (\ex^{s}-1) \gamma \bar{n}_{\rm mu}} -1\over (\ex^{s}-1) \gamma \bar{n}_{\rm b_m}}).$
 The Laplace transform of $N_m^{b_m}$, $\Phi_{N_m^{b_m}}(s)$,  satisfies the following lower and upper bounds:
\[
\Phi_{N_m^{b_m}}(s)\geq \ex^{(\ex^{-s}-1)\bar{n}_{\rm mu}} ,
\]
 and
\[
\Phi_{N_m^{b_m}}(s)\leq \ex^{(\ex^{-s}-1)\bar{n}_{\rm mu} +(\beta-1)\bar{n}_{\rm fap}}.
\]
\end{lemma}
\begin{proof}
To derive the lower bound, note that  $N_m^{b_m}\leq N_m$, and therefore, for $s\geq 0$,  $\ex^{-s N_m^{b_m}}\geq \ex^{-s N_m}$. Hence, $\E[\ex^{-s N_m^{b_m}}]\geq \E[\ex^{-s N_m}]= \ex^{\bar{n}_{\rm mu}(\ex^{-s}-1)}$.

Let $N_{\rm fap}\triangleq |\Ac_f|$. Each FAP $a_f\in\Ac_f$, at most, covers $\min(N_1^{a_f},(n_c-N_2^{a_f})^+)$ MUs, where $N_1^{a_f}\sim{\rm Poiss}(\bar{n}_{\rm mu}^{d(a_f,b_m)})$ and $N_2^{a_f}\sim{\rm Poiss}(\bar{n}_{\rm fu})$, where $\bar{n}_{\rm mu}^d$ is defined in \eqref{eq:n-bar-mu-d}. Let $N^{a_f}\triangleq \min(N_1^{a_f},(n_c-N_2^{a_f})^+)$. Conditioned on $\Ac_f$, $\{N^{a_f}\}_{a_f\in\Ac_f}$ are independent and identically distributed (i.i.d.) random variables. Then, $N_m^{b_m}\geq N_m-\sum_{a_f\in\Ac_f}N^{a_f}$.
Therefore, 
\begin{align}
\Phi_{N_m^{b_m}}(s)& = \E[\ex^{-s N_m^{b_m}}]\nonumber\\
& \leq \E[\ex^{-sN_m}]\E[\ex^{s\sum_{a_f\in\Ac_f}N^{a_f}}] \nonumber\\
&= \E[\ex^{-sN_m}]\E\Big[\E[\prod_{a_f\in\Ac_f}\ex^{sN^{a_f}}|\Ac_f]\Big]\nonumber\\
&\stackrel{(a)}{=} \E[\ex^{-sN_m}]\E[(\E[\ex^{sN^{a_f}}])^{N_{\rm fap}}]\nonumber\\
&=\ex^{(\ex^{-s}-1)\bar{n}_{\rm mu} }\ex^{\bar{n}_{\rm fap}(\E[\ex^{sN^{a_f}}]-1)},\label{eq:lemma5-1}
\end{align}
where (a) follows because conditioned on $N_{\rm fap}=i$, $\{N^{a_f}\}_{a_f\in\Ac_f}$ are $i$  i.i.d. random variables. On the other hand, 
\begin{align*}
\E[\ex^{sN^{a_f}}]\;=\; & \E[\E[\ex^{sN^{a_f}}|\ind_{N_2^{a_f}\geq n_c}]]\nonumber\\
\; = \; & \Pc(n_c,\bar{n}_{\rm fu})+\E[\ex^{sN^{a_f}}|N_2^{a_f}< n_c] (1-\Pc(n_c,\bar{n}_{\rm fu}) )\nonumber\\
\; \stackrel{(a)}{\leq} \; & \Pc(n_c,\bar{n}_{\rm fu})+\E[\ex^{sN_1^{a_f}}|N_2^{a_f}< n_c] (1-\Pc(n_c,\bar{n}_{\rm fu}) )\nonumber\\
\; \stackrel{(b)}{=} \; & \Pc(n_c,\bar{n}_{\rm fu})+\E[\ex^{sN_1^{a_f}}] (1-\Pc(n_c,\bar{n}_{\rm fu}) ),
\end{align*}
where $(a)$ holds because $N^{a_f}\leq N_1^{a_f}$ and $s\geq 0$, and $(b)$ follows from the independence of  $N_1^{a_f}$ and $N_2^{a_f}$. Also,
\begin{align*}
\E[\ex^{sN_1^{a_f}}]= \E[\E[\ex^{sN_1^{a_f}}|d(a_f,b_m)]]
=  \E[\ex^{(\ex^{s}-1)\bar{n}_{\rm mu}^{d(a_f,b_m)}}],
\end{align*}
where $\bar{n}_{\rm mu}^d$ is defined in \eqref{eq:n-bar-mu-d}. But,
\begin{align}
\E[\ex^{cd^2(a_f,b_m)}]&\;=\;\int_0^R {2r\over R^2}e^{cr^2}dr={\ex^{cR^2}-1\over cR^2}.
\end{align}
\vspace{-.5em}Therefore,
\begin{align}
\E[\ex^{sN^{a_f}}] \leq \Pc(n_c,\bar{n}_{\rm fu} )+ (1-\Pc(n_c,\bar{n}_{\rm fu}) )({\ex^{ (\ex^{s}-1) \gamma \bar{n}_{\rm mu}} -1\over (\ex^{s}-1) \gamma \bar{n}_{\rm b_m}}).\label{eq:lemma5-2}
\end{align}\vspace{-.5em}
Combining \eqref{eq:lemma5-1} and \eqref{eq:lemma5-2}  yields the desired upper bound.
\end{proof}

\section{MU served by an FAP} \label{sec:MU-by-FAP}

In this section, we analyze the outage performance of an MU serviced by an FAP, in the described  uplink network with  backhaul constraints.    We  assume that  the performance of   the users is primarily limited by the interference caused by the other users of both tiers, and therefore ignore the effect of additive Gaussian noise (AWGN) in our analysis.


Consider FAP $a_f\in\Ac_f$ at distance $d$ from MBS $b_m$, \ie $d(a_f,b_m)=d$.  Given the power control assumption, the upload SIR experienced  by user $u_m\in\Uc_m(a_f)$ in subband $i\in \{1,2,\ldots, n_s\} $ is equal to 
\begin{align}
\SIR_{m,f}&= {\frac{p_f|H^i_{u_m,a_f}|^2}{n_s} \over I_{m,f}},\label{eq:SIR-m-fap}
\end{align}
\vspace{-.5em}where 
\begin{align}
I_{m,f}=&  \sum\limits_{u_f\in \Uc_f(a_f)} \frac{p_f|H^i_{u_f,a_f}|^2}{g}+\sum\limits_{\uh_m\in \Uc_m(a_f)\backslash u_m}\frac{p_f|H^i_{\uh_m,a_f}|^2}{g}\nonumber\\
&\;+\sum\limits_{\hat{a}_f\in \Ac_f\backslash a_f }\sum\limits_{u\in \Uc_m(\hat{a}_f)\cup \Uc_f(\hat{a}_f) }\Big({d(u,\hat{a}_f)\over d(u,a_f)}\Big)^{\a}\frac{p_f|H^i_{u,\hat{a}_f}|^2}{g}\nonumber
\end{align}
\begin{align}
&\;+\sum\limits_{\uh_m\in \Uc_m{(b_m)}}({d(\uh_m,b_m)\over d(\uh_m,a_f)})^{\a}\frac{p_m|H^i_{\uh_m,a_f}|^2}{g}. \label{eq:I_macro}
\end{align}
In \eqref{eq:I_macro}, from left to right, the interference terms correspond to the interference caused by the FUs of FAP $a_f$, the other  MUs of FAP $a_f$,  users of the other FAPs and the MUs serviced by the MBS, respectively.
Given FAP $\hat{a}_f\in \Ac_f\backslash a_f $, and (femto or macro) user $u\in \Uc_m(\hat{a}_f)\cup \Uc_f(\hat{a}_f)$ covered by $\hat{a}_f$,  typically $d(u,\hat{a}_f)\ll d(u,a_f)$, or ${d(u,\hat{a}_f)\over d(u,a_f)}\ll 1$. Therefore, unless the density of FAPs is very high, the effect of the interference caused by the users of other FAPs is negligible. Under this approximation, we have 
\begin{align}
I_{m,f}=& \hspace{-2.5em} \sum\limits_{u\in \Uc_f(a_f)\; \cup \;\Uc_m(a_f)\backslash u_m}\hspace{-2em}\frac{p_f|H^i_{u,a_f}|^2}{g}+\hspace{-1em}\sum\limits_{\uh_m\in \Uc_m{(b_m)}}\hspace{-1.5em}(\delta_{\uh_m})^{\a}\frac{p_m|H^i_{\uh_m,a_f}|^2}{g}, \label{eq:I-m-f-simplified}
\end{align}
where
\begin{align}
\delta_{\uh_m}\triangleq {d(\uh_m,b_m)\over d(\uh_m,a_f)}.\label{eq:delta-u}
\end{align}
%

Define the event $\Ec=\{d(a_f,b_m)=d, N_m^{a_f}\geq 1\}$. Then MU  $u_m\in\Uc_m(a_f)$  is said to experience outage in subband $i$ if $\SIR_{m,f}$ is less than some pre-determined threshold $\theta$. Therefore, the corresponding outage probability $\P_{\rm out}^{m,f}$ of MU $u_m$ serviced by FAP $a_f$ is defined as  $\P_{\rm out}^{m,f}(\theta,d_f)= \P(\SIR_{m,f}<\theta|  \Ec)$, where $\SIR_{m,f}$ is defined in \eqref{eq:SIR-m-fap}. Since $|H^i_{u_m,a_f}|^2$ has an exponential distribution and is independent of other relevant random variables, it follows that \vspace{-.5em}
\begin{align}
P_{\rm out}^{m,f}(\theta,d_f)=1-\E[\ex^{-({\theta n_s\over \sigma^2p_f})I_{m,f}}|\Ec].\label{P-out-m-f}
\end{align}
In the following two sections, we derive analytical   upper and lower bounds on  $P_{\rm out}^{m,f}$.

Before stating the bounds, given FAP $a_f$ at distance $d_f$ from $b_m$, consider partitioning   the coverage area $\Sc_m$ of the MBS $b_m$, as described in Appendix \ref{app:A}, into $2(t+1)$ regions.  To perform this partitioning  parameters $(\kappa_0,\ldots,\kappa_t)$ are selected such that  $\kappa_0=\kappa<\kappa_1<\kappa_2<\ldots<\kappa_t=1$. For user $u$ with $\delta_u$ defined in \eqref{eq:delta-u},  $\hat{\delta}^{\rm ub}_u$ and  $\hat{\delta}^{\rm lb}_u$  are defined as follows:  $\hat{\delta}^{\rm ub}_{u}=\kappa_i^{-1}$ and $\hat{\delta}^{\rm lb}=\kappa_{i+1}^{-1}$, if $ \kappa_{i+1}^{-1}<{\delta}_u\leq \kappa_{i}^{-1}$, for  $i=0,\ldots, t-1$; $\hat{\delta}^{\rm ub}_{u}=\kappa_{i+1}$ and $\hat{\delta}^{\rm lb}_{u}=\kappa_{i}$, if $\kappa_{i}<{\delta}_u\leq \kappa_{i+1}$, for $i=0,\ldots, t-1$; and $\hat{\delta}^{\rm ub}_{u}=\kappa$ and $\hat{\delta}^{\rm lb}_{u}=0$, if $\delta_u\leq\kappa$. Note that by construction, unlike $\delta_u$, $\hat{\delta}^{\rm lb}_{u}$  and $\hat{\delta}^{\rm ub}_{u}$ are  discrete random variables. For all $u$ and $a_f$,$
\hat{\delta}^{\rm lb}_{u} \leq \delta_u\leq \hat{\delta}^{\rm ub}_{u}.$
\subsection{Upper Bound on the Outage Probability $P_{\rm out}^{m,f}$}

For $i=1,\ldots,t$, and $\hat{u}_m\in\Uc_m\backslash \Uc_m(a_f)$, let

\[
p_{i}=\P(\hat{\delta}^{\rm ub}_{\uh_m}={1\over \kappa_{i-1}})=\P(\hat{\delta}^{\rm lb}_{\uh_m}={1\over \kappa_{i}}),
\]
\[
p_{-i}=\P(\hat{\delta}^{\rm ub}_{\uh_m}=\kappa_{i})=\P(\hat{\delta}^{\rm lb}_{\uh_m}=\kappa_{i-1}),
\]
and
\[
p_0=\P(\hat{\delta}^{\rm ub}_{\uh_m}=\kappa_0)=\P(\hat{\delta}^{\rm lb}_{\uh_m}=0).
\] Also, let $
\eta\triangleq {p_f\over p_m}
$,
$\bar{n}_{m,d}\triangleq \pi(R^2-({\kappa \over 1-\kappa^2})^2d^2)\mu_m$ and
\[
q_1(\theta,d)\triangleq \sum_{i=1}^t\Big({p_i \over 1+ {\theta  \over n_h\eta\kappa_{i-1}^{\alpha}}}+{p_{-i} \over 1+ {\theta \kappa_{i}^{\alpha} \over n_h\eta }}\Big)+{p_0 \over 1+ {\theta \kappa_0^{\alpha} \over n_h\eta }}.
\]
\begin{theorem}\label{thm:1}
The outage probability of an MU serviced by an FAP located at distance $d$ from MBS, $P_{\rm out}^{m,f}(\theta,d)$, is upper bounded by \vspace{-.5em}
\begin{align*}
1\hspace{-.2em} - \hspace{-.2em} \Big(\hspace{-.1em}{1\over 1+{\theta\over n_h}}\hspace{-.2em}\Big)^{\hspace{-.2em}n_c-1}\hspace{-.2em}\ex^{\bar{n}_{m,d}( q_1\hspace{-.2em}(\theta,d) -1)} \Phi_{|\Uc_m^{\rm ns}(a_f)|}(\log(1\hspace{-.2em}+\hspace{-.2em}{\theta\over  \eta n_h\kappa_o^{\alpha}})\hspace{-.1em}),
\end{align*}
 where $\Phi_{|\Uc_m^{\rm ns}(a_f)|}$,  the Laplace transform of $|\Uc_m^{\rm ns}(a_f)|$,  is derived in  Appendix~\ref{sec:dist}.
 \end{theorem}
\begin{proof}
For MUs serviced by FAPs, as discussed in \cite{ZeinalpourJ:14-tcom},   the potential coverage area of FAP $a_f$ located at distance $d$ from $b_m$  is a circle of radius $({\kappa\over 1-\kappa^2})d$.  Due to the backhaul constraint,  all the MUs falling in this circle, $\Uc_m(a_f)$, are not serviced by $a_f$. Users in $\Uc_m(a_f)$  can be partitioned into two groups, $\Uc_m^{\rm s}(a_f)$ and $\Uc_m^{\rm ns}(a_f)$, representing the   MUs that are serviced by $a_f$ and the MUs that fall in the coverage area of   $a_f$, but are serviced by $b_m$, respectively.

Given the backhaul constraint of $n_c$ users, there are at most $n_c-1$ users (macro and femto) serviced by $a_f$  that interfere with an  FU covered by $a_f$. That is,  $|\Uc_f(a_f)\; \cup \;\Uc_m^{\rm s}(a_f)\backslash u_m|\leq n_c-1$. Also, we always have $\Uc_m(b_m)\subseteq\Uc_m\backslash\Uc_m^{\rm s}(a_f)$. Therefore, from \eqref{eq:I_macro},

\begin{align}
I_{m,f} &\leq  \sum\limits_{\ell=1}^{n_c-1}{p_f\over g}|H_{\ell}|^2\;+\sum\limits_{\uh_m\in \Uc_m\backslash \Uc_m^{\rm s}(a_f)}{ (p_m\delta_{\uh_m})^{\a} \over g}|H^i_{\uh_m,a_f}|^2\nonumber\\
 &=  \sum\limits_{\ell=1}^{n_c-1}{p_f\over g}|H_{\ell}|^2\;+\sum\limits_{\uh_m\in \Uc_m\backslash \Uc_m(a_f)}{ p_m(\delta_{\uh_m})^{\a} \over g}|H^i_{\uh_m,a_f}|^2\nonumber\\
 &\hspace{1em}+ \sum\limits_{\uh_m\in \Uc_m^{\rm ns}(a_f) }{p_m (\delta_{\uh_m})^{\a} \over g}|H^i_{\uh_m,a_f}|^2\nonumber\\
& \stackrel{(a)}{\leq} \sum\limits_{\ell=1}^{n_c-1}{p_f\over g}|H_{\ell}|^2\;+\sum\limits_{\uh_m\in \Uc_m\backslash \Uc_m(a_f)}(\delta_{\uh_m})^{\a}{p_m\over g}|H^i_{\uh_m,a_f}|^2\nonumber\\
 &\hspace{1em}+ \sum\limits_{\uh_m\in \Uc_m^{\rm ns}(a_f) }{p_m\over \kappa_o^{\a}g}|H^i_{\uh_m,a_f}|^2\nonumber\\
& \stackrel{(b)}{\leq} \sum\limits_{\ell=1}^{n_c-1}\frac{p_f}{g}|H_{\ell}|^2\;+\sum\limits_{\uh_m\in \Uc_m\backslash \Uc_m(a_f)}{p_m(\hat{\delta}^{\rm ub}_{\uh_m})^{\a}\over g}|H^i_{\uh_m,a_f}|^2\nonumber\\
 &\hspace{1em}+ \sum\limits_{\uh_m\in \Uc_m^{\rm ns}(a_f) }{p_m\over \kappa_o^{\a}g}|H^i_{\uh_m,a_f}|^2,\label{eq:i-m-f-UB-step3}
\end{align}
where $\{|H_{\ell}|^2: \ell=1,\ldots,n_c-1\}$ are i.i.d. exponential random variables independent of other random variables in \eqref{eq:i-m-f-UB-step3}. Also, $(a)$ holds because by assumption, ${d(\uh_m,b_m)/ d(\uh_m,\hat{a}_f)}\leq \kappa_o^{-1}$, for all $\uh_m\in\Uc_m$, and all $\hat{a}_f\in\Ac_f$, and $(b)$ follows because $\delta_u\leq \hat{\delta}_u^{\rm ub}$.

Since the MUs in $\Uc_m$ are generated according to a PPP and the users in  $ \Uc_m\backslash \Uc_m(a_f)$ and $ \Uc_m^{\rm ns}(a_f)$ have non-overlapping supports, they are independent. Therefore, combining \eqref{P-out-m-f} and \eqref{eq:i-m-f-UB-step3}, it follows that
\begin{align}
P_{\rm out}^{m,f}=& 1-\E[\ex^{-({\theta n_s\over \sigma^2p_f})I_{m,f}}|\Ec]\nonumber\\
\leq &1 -  \Big({1\over 1+{\theta\over n_h}}\Big)^{n_c-1}\nonumber\\
&\hspace{1.5em}\times \E\hspace{-.2em}\left[\hspace{-.2em}\Big(\hspace{-.2em}\E\hspace{-.2em}\Big[\ex^{-{\theta \over n_h\eta\sigma^2}(\hat{\delta}^{\rm ub}_{\uh_m})^{\alpha}|H^i_{\uh_m,a_f}|^2}\Big]\hspace{-.2em}\Big)^{|\Uc_m|-|\Uc_m(a_f)|}\hspace{-.1em}\right]\nonumber\\
&\hspace{1.5em}\times \E\hspace{-.2em}\Big[\hspace{-.2em}\Big(\hspace{-.2em}{1\over 1+{\theta\over  \eta n_h\kappa_o^{\alpha}}}\hspace{-.2em}\Big)^{|\Uc_m^{\rm ns}(a_f)|}\hspace{-.1em}\Big].\label{eq:Pout-UB-step1}
\end{align}
Since $\hat{\delta}^{\rm ub}_{\uh_m}$ and $|H^i_{\uh_m,a_f}|$ are independent,
\begin{align}
&\E\Big[\ex^{-{\theta \over n_h\eta\sigma^2}(\hat{\delta}^{\rm ub}_{\uh_m})^{\alpha}|H^i_{\uh_m,a_f}|^2}\Big]=
 q_1(\theta,d).\label{eq:Exp-q1}
\end{align}
Finally, $|\Uc_m|-|\Uc_m(a_f)|$ is a Poisson random variable of mean $\bar{n}_{m,d}$. Therefore, combining \eqref{eq:Pout-UB-step1} and \eqref{eq:Exp-q1} yields the desired result.
\end{proof}

\subsection{Lower Bound on the Outage Probability $P_{\rm out}^{m,f}$}\label{sec:pout-LB}
Consider partitioning the MUs in $\Uc_m{(b_m)}\backslash \Uc_m^{\rm ns}(a_f)$ into two groups:
\begin{enumerate}
\item[i)]$\Uc_m^{\rm in}(b_m)$: the subset of MUs that  fall into the coverage area of at least one FAP in $\Ac_f\backslash a_f$, but are serviced by the MBS due to the backhaul constraints, \ie
\[
\Uc_m^{\rm in}(b_m)\triangleq \cup_{\hat{a}_f\in\Ac_f\backslash a_f}\Uc_m^{\rm ns}(\hat{a}_f),
\]
\item[ii)] $\Uc_m^{\rm out}(b_m)$: the subset of MUs that are serviced by the MBS because they do not fall into the coverage area of any FAP, \ie 
\[
\Uc_m^{\rm out}(b_m)\triangleq \Uc_m(b_m) \backslash (\Uc_m^{\rm in}(b_m) \cup\Uc_m^{\rm ns}(a_f)).
\]
\end{enumerate}

For $i=1,\ldots,t$, and $\hat{u}_m\in\Uc_m^{\rm out}(b_m)$, let
\[
p'_{i}=\P(\hat{\delta}^{\rm ub}_{\uh_m}={1\over \kappa_{i-1}})=\P(\hat{\delta}^{\rm lb}_{\uh_m}={1\over \kappa_{i}}),
\]
\[
p'_{-i}=\P(\hat{\delta}^{\rm ub}_{\uh_m}=\kappa_{i})=\P(\hat{\delta}^{\rm lb}_{\uh_m}=\kappa_{i-1}),
\]
and   $p'_0=\P(\hat{\delta}^{\rm ub}_{\uh_m}=\kappa_0)=\P(\hat{\delta}^{\rm lb}_{\uh_m}=0).$
Now define 
\begin{align}
q_2(\theta,d)&\triangleq p'_0+\sum_{i=1}^t\Big({p'_i \over 1+ {\theta  \over n_h\eta\kappa_{i}^{\alpha}}}+{p'_{-i} \over 1+ {\theta \kappa_{i-1}^{\alpha} \over n_h\eta }}\Big),
\end{align}  $\gamma_1\triangleq \pi (1-q_2(\theta,d))({\kappa \over 1-\kappa^2})^2 \mu_m,$
 $
\gamma_2\triangleq{\theta \over\eta n_h (1+\kappa)^{\alpha}},
$
and
$
\gamma_3\triangleq \pi\mu_m(({\kappa\over 1-\kappa^2})^2-({\kappa_o\over 1-\kappa_o^2})^2).
$
Consider FAP $a_f$ at distance $d$ from MBS $b_m$ and FAP $\hat{a}_f\in\Ac_f\backslash a_f$. Let $(D_1,D_2)=(d(\hat{a}_f,b_m),d(\hat{a}_f,a_f))$, and define
\begin{align}
\gamma_4\triangleq \E\Big[\ex^{D_1^2(\gamma_1 -{\gamma_2\gamma_3 D_1^{\alpha}\over D_2^{\alpha}+\gamma_2D_1^{\alpha}})} \Big].
\end{align}
Note that $\gamma_4$ can easily be computed through Monte Carlo simulations.

\begin{theorem}\label{thm:2}
Let
\[
\chi\triangleq \Big({1-\ex^{-\gamma_1R^2}\over \gamma_1 R^2}\Big)(1-\Pc(n_c,\bar{n}_{\rm fu}))+\gamma_4\Pc(n_c,\bar{n}_{\rm fu}).
\]
Then, $P_{\rm out}^{m,f}(\theta,d)$, the outage probability of an MU serviced by an FAP located at distance $d$ from the MBS, is lower bounded by 
 \begin{align}
&1-  \ex^{-\bar{n}_{\rm mu}(1-q_2)}  {\ex^{\bar{n}_{\rm fap}(\chi-1)}-\ex^{-\bar{n}_{\rm fap}}\over  (1-\ex^{-\bar{n}_{\rm fap}})\chi}
\nonumber\\
& \hspace{3em}\times \Phi_{|\Uc_m^{\rm ns}(a_f)|}\Big(\log\Big(1+{\theta\over  \eta n_h\kappa^{\alpha}}\Big)\Big)(1+O(\kappa^{\alpha})),\label{eq:i-m-f-LB-final}
\end{align}
where $\Phi_{|\Uc_m^{\rm ns}(a_f)|}$ is computed in Section~\ref{sec:dist}.
\end{theorem}
\begin{proof}
 Considering the described  partitioning of users in $\Uc_m{(b_m)}\backslash \Uc_m^{\rm ns}(a_f)$, and ignoring the interference caused by the other FUs and MUs that are serviced by $a_f$, $I_{m,f}$ can be lower bounded as 
\begin{align}
I_{m,f}\geq &\sum \limits_{\uh_m\in \Uc_m^{\rm in}(b_m)}{p_m\over g}(\delta_{\uh_m})^{\a}|H^i_{\uh_m,a_f}|^2\nonumber\\
\;&+\sum\limits_{\uh_m\in \Uc_m^{\rm out}(b_m)}{p_m\over g}(\delta_{\uh_m})^{\a}|H^i_{\uh_m,a_f}|^2\nonumber\\
\;&+ \sum\limits_{\uh_m\in \Uc_m^{\rm ns}(a_f)}{p_m\over g}(\delta_{\uh_m})^{\a}|H^i_{\uh_m,a_f}|^2. \label{eq:i-m-f-LB-step1.2}
\end{align}

For users in $\Uc_m^{\rm in}(b_m)$, consider FAP $\hat{a}_f\in\Ac_f \backslash a_f$,   and user $\uh_m\in\Uc_m^{\rm ns}(\hat{a}_f)$. (Refer to Fig.~\ref{fig:dist-ratio}.) Let  $d_o=d(\hat{a}_f,b_m)$. If FAP $a_f$ does not fall into the coverage area of $\hat{a}_f$, as shown in Fig.~\ref{fig:dist-ratio}, $d(\uh_m,b_m)\geq {1\over 1-\kappa^2}d_o-{\kappa\over 1-\kappa^2}d_o={d_o\over 1+\kappa}$ and $d(\uh_m,a_f)\leq d(\hat{c}_f,a_f)+{\kappa\over 1-\kappa^2}d_o$, where $\hat{c}_f$ denotes the center of the coverage area of $\hat{a}_f$. Hence,
\begin{align}
{d(\uh_m,b_m)\over d(\uh_m,a_f)}&\geq {{1\over 1+\kappa}d(\hat{a}_f,b_m) \over d(\hat{c}_f,a_f)+{\kappa\over 1-\kappa^2}d(\hat{a}_f,b_m)}.
\end{align}
On the other hand, since both $\hat{a}_f$ and $\hat{u}_m$ are  located in a circle of radius ${\kappa\over 1-\kappa^2}d_o$, $d(\hat{a}_f,\hat{u}_m)\leq {2\kappa\over 1-\kappa^2}d_o$. Therefore $d(\hat{a}_f,\hat{u}_m)=O(\kappa)$, and 
\[
{{1\over 1+\kappa}d(\hat{a}_f,b_m) \over d(\hat{c}_f,a_f)+{\kappa\over 1-\kappa^2}d(\hat{a}_f,b_m)}={{d(\hat{a}_f,b_m) \over (1+\kappa)d(\hat{a}_f,a_f)}}+O(\kappa).
\]
For users in $\Uc_m^{\rm ns}(a_f)$,  ${1\over \kappa}\leq {d(\uh_m,b_m)\over d(\uh_m,a_f)}\leq {1\over \kappa_o}$.  Let  $\delta_{\hat{a}_f}\triangleq {{d(\hat{a}_f,b_m) \over d(\hat{a}_f,a_f)}}$. Then, noting that $\hat{\delta}^{\rm lb}_{\uh_m}\leq \delta_{\uh_m}$, from \eqref{eq:i-m-f-LB-step1.2}, conditioned on the event that none of the other FAPs falls  into the coverage area of $a_f$, it follows that
\begin{align}
I_{m,f}\geq
&\sum\limits_{\uh_m\in \Uc_m^{\rm ns}(a_f)}{p_m\over \kappa^{\alpha} g}|H^i_{\uh_m,a_f}|^2\nonumber\\
& \;+\sum_{\hat{a}_f\in\Ac_f\backslash a_f}{p_m\over g(1+\kappa)^{\alpha}} (\delta_{\hat{a}_f})^{\alpha}\sum_{\hat{u}_m\in \Uc_m^{\rm ns}(\hat{a}_f)}|H^i_{\uh_m,a_f}|^2\nonumber\\
& \;+\sum\limits_{\uh_m\in \Uc_m^{\rm out}{(b_m)}}{(\hat{\delta}^{\rm lb}_{\uh_m})^{\a}p_m\over g}|H^i_{\uh_m,a_f}|^2+{O(\kappa^{\alpha})}. \label{eq:i-m-f-LB-step3}
\end{align}
\begin{figure}[t]
\begin{center}
\includegraphics[width=6.5cm]{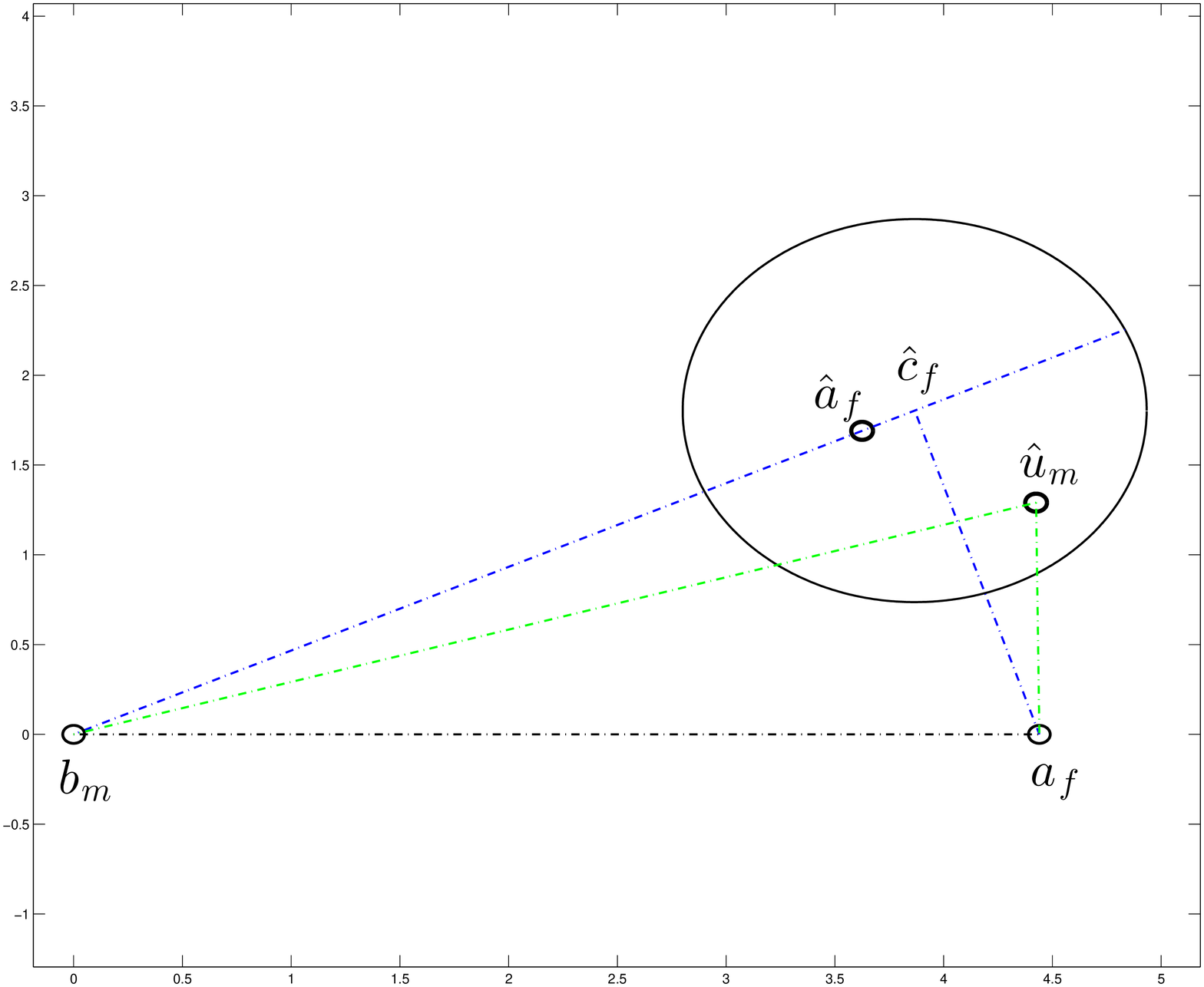}\vspace{-.5em}\caption{User $\uh_m\in\Uc_m^{\rm ns}(\hat{a}_f)$.}\label{fig:dist-ratio}\vspace{-2em}
\end{center}\vspace{-1.5em}
\end{figure}

All the interference terms in \eqref{eq:i-m-f-LB-step3} have non-overlapping supports, and hence, conditioned on the locations of FPAs, are independent.  Therefore, combining \eqref{P-out-m-f} and \eqref{eq:i-m-f-LB-step3}, it follows that\vspace{-1em}
\begin{align}
P_{\rm out}^{m,f}&= 1-\E[\ex^{-({\theta n_s\over \sigma^2p_f})I_{m,f}}|\Ac_f,\Ec]\nonumber\\
&\geq  1-\E\Big[\ex^{-{\theta \over \eta n_h\sigma^2\kappa^{\alpha}}\sum\limits_{\hat{u}_m\in\Uc_m^{\rm ns}(a_f)}|H^i_{\hat{u}_m,a_f}|^2 }\Big|\Ec\Big]\nonumber\\
&\;\;\;\times \hspace{-.2em}\E\Big[\E\Big[\ex^{-{\theta \over\eta n_h\sigma^2 }\sum\limits_{\uh_m\in \Uc_m^{\rm out}(b_m)}(\hat{\delta}^{\rm lb}_{\uh_m})^{\a} |H_{\hat{u}_m,a_f}^i|^2}\Big|\Ec,\Ac_f\Big]\nonumber\\
&\;\;\;\times \hspace{-.2em}\E \hspace{-.2em}\Big[\hspace{-.1em}\ex^{-{\theta \over\eta n_h\sigma^2 (1+\kappa)^{\alpha} }\hspace{-.4em}\sum\limits_{\hat{a}_f\in\Ac_f\backslash a_f}\hspace{-1em}(\delta_{\hat{a}_f})^{\a} \hspace{-1em} \sum\limits_{\hat{u}_m\in\Uc_m^{\rm ns}(\hat{a}_f)}\hspace{-.5em}|H^i_{\hat{u}_m,a_f}|^2}\hspace{-.2em}\Big|\Ec,\Ac_f\hspace{-.1em}\Big]\hspace{-.1em}\Big]\nonumber\\
&\;\;\;\times\hspace{-.2em}(1+ O(\kappa^{\alpha})). \label{eq:i-m-f-LB-step4}
\end{align}
Let $S_{\rm out}$ denote the area of the region that is not covered by any of the FAPs. Then, conditioned on $(\Ec,\Ac_f)$, $|\Uc_m^{\rm out}(b_m)|$ is distributed as ${\rm Poiss}(S_{\rm out}\mu_m)$. Therefore, as
\[
\E\Big[\ex^{-{\theta \over\eta n_h\sigma^2 }(\hat{\delta}^{\rm lb}_{\uh_m})^{\a} |H_{\hat{u}_m,a_f}^i|^2}\Big|\uh_m\in\Uc_m^{\rm out}(b_m)\Big]=q_2(\theta,d),
\]
it follows that
\begin{align}
&\E\Big[\ex^{-{\theta \over\eta n_h\sigma^2 }\sum\limits_{\uh_m\in \Uc_m^{\rm out}(b_m)}(\hat{\delta}^{\rm lb}_{\uh_m})^{\a} |H_{\hat{u}_m,a_f}^i|^2}\Big|\Ec,\Ac_f\Big]\nonumber\\
&=
\E\Big[\Big(\E[\ex^{-{\theta \over\eta n_h\sigma^2 }(\hat{\delta}^{\rm lb}_{\uh_m})^{\a} |H_{\hat{u}_m,a_f}^i|^2}]\Big)^{|\Uc_m^{\rm out}(b_m)|}\Big|\Ec,\Ac_f\Big]\nonumber\\
&=\ex^{(q_2(\theta,d)-1)S_{\rm out}\mu_m}.\label{eq:i-m-f-LB-step4-1}
\end{align}
On the other hand,
\begin{align}
&\E\Big[\ex^{-{\theta \over\eta n_h\sigma^2 (1+\kappa)^{\alpha} }\sum\limits_{\hat{a}_f\in\Ac_f\backslash a_f}(\delta_{\hat{a}_f})^{\a} \sum\limits_{\hat{u}_m\in\Uc_m^{\rm ns}(\hat{a}_f)}|H^i_{\hat{u}_m,a_f}|^2}\Big|\Ec,\Ac_f\Big]\nonumber\\
&=\prod_{\hat{a}_f\in\Ac_f\backslash a_f}\Big({1\over 1+ \gamma_2(\delta_{\hat{a}_f})^{\alpha} }\Big)^{N_{m,nc}^{\hat{a}_f}}.\label{eq:i-m-f-LB-step4-2}
\end{align}
Combining \eqref{eq:i-m-f-LB-step4}, \eqref{eq:i-m-f-LB-step4-1} and \eqref{eq:i-m-f-LB-step4-2}, and noting that $S_{\rm out}\geq \pi R^2-\pi ({\kappa \over 1-\kappa^2})^2 \sum_{\hat{a}_f\in\Ac_f\backslash a_f}d^2(a_f,b_m)$, we have
\begin{align}
P_{\rm out}^{m,f}& \hspace{-.2em}\geq  \hspace{-.2em} 1  \hspace{-.2em} -  \hspace{-.2em} \ex^{-\pi R^2\mu_m(1-q_2)}  \hspace{-.2em}\E \hspace{-.2em}\Big[\ex^{-{\theta \over \eta n_h\sigma^2\kappa^{\alpha}} \hspace{-.3em}\sum\limits_{\hat{u}_m\in\Uc_m^{\rm ns}(a_f)} \hspace{-1em}|H^i_{\hat{u}_m,a_f}|^2  } \hspace{-.2em}\Big|\Ec\Big] \nonumber\\
&\hspace{-2em}\times  \hspace{-.2em}\E \hspace{-.2em}\Big[  \hspace{-1em}\prod_{\hat{a}_f\in\Ac_f\backslash a_f} \hspace{-1em}\ex^{\pi (1-q_2)({\kappa \over 1-\kappa^2})^2d^2(\hat{a}_f,b_m) \mu_m}\Big( \hspace{-.2em}{1\over 1+{\gamma_2 (\delta_{\hat{a}_f})^{\alpha}}} \hspace{-.2em}\Big)^{N_{m,nc}^{\hat{a}_f}} \hspace{-.2em}\Big|\Ec \hspace{-.1em}\Big]\nonumber\\
&\hspace{-.2em}= 1-\ex^{-\bar{n}_{\rm mu}(1-q_2)}\E\Big[ \Big({1\over 1+{\theta \over \eta n_h \kappa^{\alpha}}}\Big)^{|\Uc_m^{\rm ns}(a_f)|}\Big|\Ec\Big]\nonumber\\
&\times  \E\hspace{-.1em}\Big[\hspace{-.1em}\Big(\hspace{-.1em}\E[\ex^{\gamma_1d^2(\hat{a}_f,b_m) }\Big({1\over 1+{\gamma_2(\delta_{\hat{a}_f})^{\alpha}}}\Big)^{N_{m,nc}^{\hat{a}_f}}]\Big)^{|\Ac_f|-1}\Big|\hspace{-.1em}\Ec\Big]\hspace{-.1em}.\label{eq:i-m-f-LB-step5}
\end{align}

Let $(D_1,D_2)=(d(\hat{a}_f,b_m),d(\hat{a}_f,a_f))$. Employing the upper bound derived in Lemma \ref{lemma:1}, we have
\begin{align}
&\E\Big[\ex^{\gamma_1 D_1^2 }\Big({1\over 1+{\gamma_2(\delta_{\hat{a}_f})^{\alpha} }}\Big)^{N_{m,nc}^{\hat{a}_f}}\Big]\nonumber\\
&=\E\Big[\E\Big[\ex^{\gamma_1D_1^2 }\Big({1\over 1+{\gamma_2(\delta_{\hat{a}_f})^{\alpha} }}\Big)^{N_{m,nc}^{\hat{a}_f}}\Big|D_1,D_2\Big]\Big]\nonumber\\
&\leq \hspace{-.2em}{(1 \hspace{-.2em}- \hspace{-.2em}\Pc(n_c,\bar{n}_{\rm fu})) \E[\ex^{\gamma_1D_1^2 }]\hspace{-.2em} + \hspace{-.2em}\Pc(n_c,\bar{n}_{\rm fu}) \E[\ex^{D_1^2(\gamma_1 -{\gamma_2\gamma_3 D_1^{\alpha}\over D_2^{\alpha}+\gamma_2D_1^{\alpha}})} ]}\nonumber\\
&=\chi.\label{eq:i-m-f-LB-step6}
\end{align}
 Finally, combining \eqref{eq:i-m-f-LB-step5} and \eqref{eq:i-m-f-LB-step6} yields the desired result.
\end{proof}

\section{MU served by the MBS} \label{sec:MU-by-MBS}
In this section, we analyze the outage performance of an MU serviced by the MBS. The upload SIR experienced  by user $u_m\in\Uc_m(b_m)$ in subband $i\in \{1,2,\ldots, n_s\} $ is equal to
\begin{align}
\SIR_{m,m}&= {\frac{p_m|H^i_{u_m,b_m}|^2}{n_s} \over I_{m,m}},\label{eq:SIR-m-bm}
\end{align}
where 
\begin{align*}
I_{m,m}=& \sum\limits_{a_f\in \Ac_f  }\sum\limits_{u\in \Uc_m({a}_f)\cup\; \Uc_f({a}_f) }\Big({d(u,{a}_f)\over d(u,b_m)}\Big)^{\a}\frac{p_f|H^i_{u,{a}_f}|^2}{g}\nonumber\\
&+\sum\limits_{\uh_m\in \Uc_m{(b_m)\backslash u_m}}\frac{p_m|H^i_{\uh_m,b_m}|^2}{g}.
\end{align*}
According to the assumed access policy, for user $u\in\Uc_m(a_f)$, we  have $d(u,a_f)\leq \kappa d(u,b_m)$, and therefore $({d(u,{a}_f)\over d(u,b_m)})^{\a}\leq\kappa^{\alpha}\ll1$. Also, for user $u\in\Uc_f(a_f)$, it is reasonable to assume that $d(u,{a}_f) \ll d(u,b_m)$. Hence, the first interference term in \eqref{eq:I-m-m} is negligible compared to the second one.  Under this approximation, \vspace{-.5em}
\begin{align}
I_{m,m} = \sum\limits_{\uh_m\in \Uc_m{(b_m)\backslash u_m}}\frac{p_m|H^i_{\uh_m,b_m}|^2}{g}.\label{eq:I-m-m}
\end{align}

\begin{theorem}\label{thm:3}
Let $\bar{n}_o\triangleq \bar{n}_{\rm b_m}({\kappa \over 1-\kappa^2})^2$ and $\e  \triangleq \ex^{-\bar{n}_{b_m}+\bar{n}_{\rm fap}({\ex^{\bar{n}_o}-1\over \bar{n}_o}-1)}$. The outage probability experienced by an MU serviced by the MBS, $\P_{\rm out}^{m,m}(\theta)=\P(\SIR_{m,m}\leq \theta)$, satisfies
\[
\P_{\rm out}^{m,m}(\theta) \geq 1- {1+{\theta \over n_h}\over 1-\e}\Phi_{N_m^{b_m}}(\ln(1+{\theta\over n_h})),
\]
\[
 \P_{\rm out}^{m,m}(\theta)\leq 1- (1+{\theta \over n_h})(\Phi_{N_m^{b_m}}(\ln(1+{\theta\over n_h}))-\e).
\]
\end{theorem}
\begin{proof}
Combining \eqref{eq:SIR-m-bm} and \eqref{eq:I-m-m}, since $|H^i_{u_m,b_m}|^2$ satisfies an exponential distribution, we have
\begin{align}
\P_{\rm out}^{m,m}(\theta) &= 1-\E[\ex^{-({\theta n_s\over \sigma^2p_m})I_{m,m}}]\nonumber\\
&= 1-\E\Big[ ({1\over 1+ {\theta \over n_h}})^{N_m^{b_m}-1} \Big| N_m^{b_m}\geq 1 \Big].
\end{align}
Let $a\triangleq {1\over 1+ {\theta \over n_h}}$. Then,

\begin{align*}
\E\Big[ ({1\over 1+ {\theta \over n_h}})^{N_m^{b_m}-1} \Big| N_m^{b_m}\geq 1 \Big]&=\E[ a^{N_m^{b_m}-1} | N_m^{b_m}\geq 1 ]\nonumber\\
&\hspace{-4em}=\sum_{i=1}^{\infty}a^{i-1} \P(N_m^{b_m}=i|N_m^{b_m}\geq 1)\nonumber\\
&\hspace{-4em}=\sum_{i=1}^{\infty}a^{i-1}{ \P(N_m^{b_m}=i)\over \P(N_m^{b_m}\geq 1)}\nonumber\\
&\hspace{-4em}={a^{-1}(\E[a^{N_m^{b_m}}]-\P(N_m^{b_m}=0))\over 1- \P(N_m^{b_m}=0)}\nonumber\\
&\hspace{-4em}={a^{-1}(\Phi_{N_m^{b_m}}(-\ln a)-\P(N_m^{b_m}=0))\over 1- \P(N_m^{b_m}=0)}.
\end{align*}
We first derive an upper bound in $\P(N_m^{b_m}=0)$. As defined in Section \ref{sec:pout-LB}, let $\Uc^{\rm out}_m(b_m)$ denote the set of users in $\Uc_m(b_m)$ that fall into the coverage area of no FAP. Also, let $\Uc^{\rm in}_m(b_m)=\Uc_m(b_m)\backslash \Uc^{\rm out}_m(b_m)$. Then,
\begin{align*}
\hspace{-.2em}\P(\hspace{-.2em}N_m^{b_m}\hspace{-.2em}=\hspace{-.2em}0\hspace{-.1em})&\hspace{-.2em}=\hspace{-.2em}\P(|\Uc^{\rm out}_m\hspace{-.1em}(\hspace{-.1em}b_m)|\hspace{-.2em}=\hspace{-.2em}|\Uc^{\rm in}_m\hspace{-.1em}(\hspace{-.1em}b_m)|\hspace{-.2em}=\hspace{-.2em}0) \hspace{-.2em}\leq \hspace{-.2em}\P(|\Uc^{\rm out}_m\hspace{-.1em}(\hspace{-.1em}b_m)|\hspace{-.2em}=\hspace{-.2em}0\hspace{-.1em}).
\end{align*}
Conditioned on $\Ac_f$, $|\Uc^{\rm out}_m(b_m)|$ is distributed as ${\rm Poiss}(S_{\rm out}\mu_m)$. Therefore,
\begin{align*}
 \P(|\Uc^{\rm out}_m(b_m)|=0)=\E[\ex^{-S_{\rm out}\mu_m}].
\end{align*}
But $S_{\rm out} \geq \pi R^2-\pi({\kappa \over 1-\kappa^2})^2\sum_{a_f\in\Ac_f} d^2(a_f,b_m)$. Hence,
\begin{align*}
 \P(|\Uc^{\rm out}_m(b_m)|\hspace{-.2em}=\hspace{-.2em}0)&\leq \E[\ex^{-\bar{n}_{b_m}+\pi({\kappa \over 1-\kappa^2})^2\mu_m\sum_{a_f\in\Ac_f} d^2(a_f,b_m)}]\nonumber\\
 &= \ex^{-\bar{n}_{b_m}}\E\Big[(\E[\ex^{\pi({\kappa \over 1-\kappa^2})^2d^2(a_f,b_m)}])^{|\Ac_f|}\Big]\nonumber\\
  &= \ex^{-\bar{n}_{b_m}}\E\Big[\Big({\ex^{\bar{n}_o}-1\over \bar{n}_o}\Big)^{|\Ac_f|}\Big]\nonumber\\
  &= \ex^{-\bar{n}_{b_m}+\bar{n}_{\rm fap}({\ex^{\bar{n}_o}-1\over \bar{n}_o}-1)}\nonumber\\
  &=\e.
\end{align*}
\end{proof}
\begin{remark}
Combining the upper and lower bounds on $\Phi_{N_m^{b_m}}(.)$ derived in  Lemma \ref{lemma:2} with the lower and upper bounds of Theorem \ref{thm:3} yields lower and upper bounds on $\P_{\rm out}^{m,m}(\theta,d_f) $, respectively, which are in terms of the system parameters.
\end{remark}

\section{Numerical results}\label{sec:num}

In this section, we present some simulation results and compare the results with the  obtained upper and lower bounds.  Throughout this section, the  simulation results are generated by $10^5-10^6$ realizations. We also compare our results with the bounds derived in~\cite{ZeinalpourJ:14-tcom} for the case in which there is no backhaul constraint.  The considered setup  is a two-tier network in a circle of radius $R=1\, {\rm Km}$ with the MBS located at the center. In the ensuing plots, unless otherwise stated,  the default values in Table \ref{Tab.1} are used.

To evaluate the upper and lower bounds stated in Theorems \ref{thm:1} and \ref{thm:2}, we need to compute the values of $\{p_i\}_{i=-t}^{i=t}$ and $\{p'_i\}_{i=-t}^{i=t}$, respectively. The values of $\{p_i\}$ are given in Lemma 1 of \cite{ZeinalpourJ:14-tcom}. As discussed in Appendix \ref{app:A}, for small values of $\kappa$, the MUs in $\Uc_m^{\rm out}(b_m)$ have a near-uniform distribution. Therefore, the same  Lemma 1 from \cite{ZeinalpourJ:14-tcom} also provides a reasonable approximation for the values of $\{p'_i\}$.

\begin{table}[t]
\footnotesize
\begin{center}
\caption{Simulation parameters}\label{Tab.1}
\begin{tabular}{|c|c|c|}
  \hline
  Sym.         & Description                              & Default Values \\
\hline \hline
 $\lambda_f$        &   density of FAPs  &  $5\times 10^{-6} \, {\rm m^{-2}}$   \\
\hline
$\mu_f$        &   density of femto users  &  $5\times 10^{-3} \, {\rm m^{-2}}$   \\
\hline
 $\mu_m$            &   density of macrocell users     &    $40\times 10^{-6}\, {\rm m^{-2}}$      \\
\hline
  $\delta$             &    ring width of FUs placement   &    $5\, {\rm m}$     \\
 \hline
 $r_f$             &ring  internal radius of  FUs placement&    $10\, {\rm m}$     \\
\hline
$\alpha$             &     path loss exponent  &   4     \\
\hline
$T$             &     SIR threshold level  &    $2$     \\
\hline
$n_s$             &   number of subbands    &    $32$     \\
\hline
$n_h$             &   number of subchannels in each subbands    &    $1024$     \\
\hline
$\eta$             &    power ratio between FAPs and MBS   &    $40$     \\
\hline
$\kappa$             &    handover parameter   &    $0.08$     \\
\hline
    \end{tabular}
\end{center}
\end{table}

\begin{figure}
\begin{center}
\includegraphics[height=5.5cm]{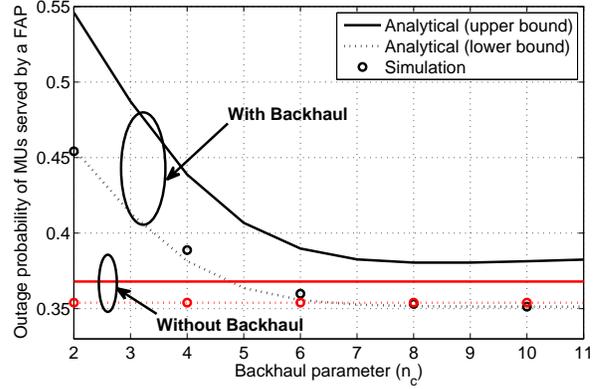}\caption{Outage probability of an MU served by an FAP (at distance of $800\, {\rm m\;}$ from the MBS) as
a function backhaul parameter $n_c$.}\label{versus_nc1}\end{center}
\end{figure}

Fig.~\ref{versus_nc1} shows the effect of the backhaul capacity $n_c$ on the outage probability experienced by the MUs  serviced by a FAP located at $d_f=800\, {\rm m\;}$ from $b_m$.  Increasing the backhaul capacity $n_c$ results in statistically  more MUs being serviced by FAPs, which in turn reduces the cross-tier interference experienced  by users served by the FAPs. At the same time, this will increase the co-tier interference.  However,  from the figure, the cross-tier interference is the dominant term  compared to the co-tier one. Also, it can be observed that as $n_c$ increases, the backhaul-constraint bounds converge to those of without restriction, computed in~\cite{ZeinalpourJ:14-tcom}. It should be mentioned that for all values of $n_c$, the bounds are consistent with  the simulation results, which confirm the accuracy of the derived analytical bounds.

\begin{figure}[t]
\begin{center}
\includegraphics[height=6cm]{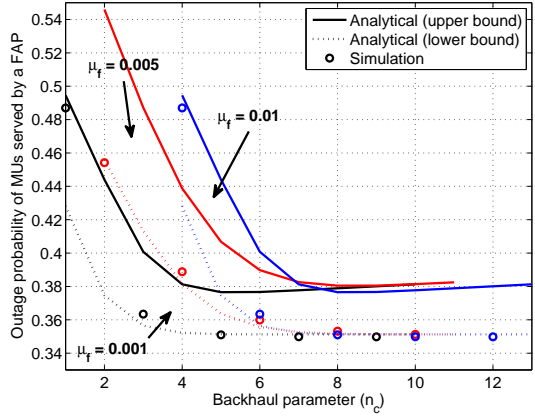}\caption{Outage probability of an MU served by an FAP as
a function backhaul parameter $n_c$ for different FUs densities.}\label{versus_nc2}\end{center}
\end{figure}

Fig.~\ref{versus_nc2} shows  the outage probability experienced by the MUs serviced by an FAP located at $d_f=800\, {\rm m\;}$ from $b_m$ as a function of the backhaul parameter $n_c$, for different values of $\mu_f$, the FUs' density. Obviously, as $\mu_f$ increases, the interference caused by FUs also increases. This will increase the outage probability of the MUs serviced by the FAPs. For large values of $n_c$, the effect of backhaul constraint fades away, and since the dominant cross-tier interference does not depend on $\mu_f$, the curves converge together.

\begin{figure}[t]
\begin{center}
\subfigure[]{\includegraphics[height=45mm]{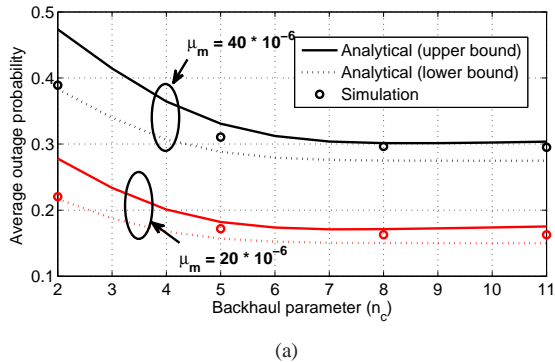}\label{vs_nc_mum1}}
\subfigure[]{\includegraphics[height=45mm]{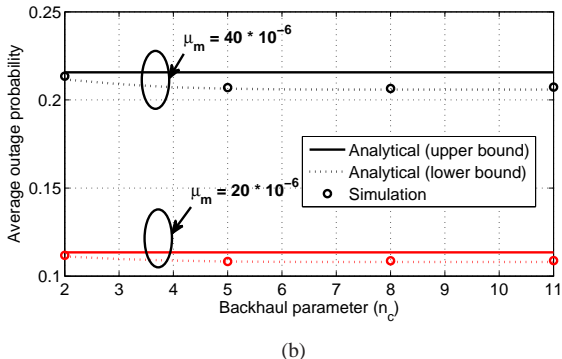}\label{vs_nc_mum1}}
\end{center}
\begin{center}
\caption{Outage probability of MUs as
a function of the backhaul parameter $n_c$ for different MUs densities a) MUs served by FAPs  b) MUs served by the MBS.} \label{vs_nc_mum}
\end{center}
\end{figure}

 Fig.~\ref{vs_nc_mum} shows the average outage probability experienced by MUs  as a function of  $n_c$,  for two different values of $\mu_m$, the MUs' density.  Increasing $\mu$ increases both cross- and co-tiers interferences, and hence results in higher outage probabilities.\footnote{For plotting the average outage probability experienced by  MUs served by the FAPs, we take the expected values of the upper and lower bounds obtained in Theorems \ref{thm:1} and \ref{thm:2} by considering the randomness in $d_f$.}

\begin{figure}
\begin{center}
\subfigure[]{\includegraphics[height=55mm]{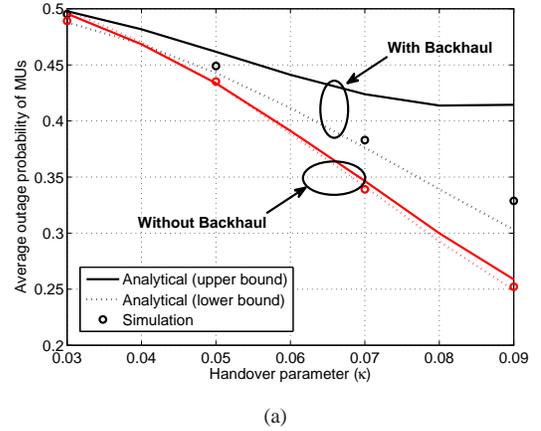}\label{vs_kappa_FBS}}
\subfigure[]{\includegraphics[height=55mm]{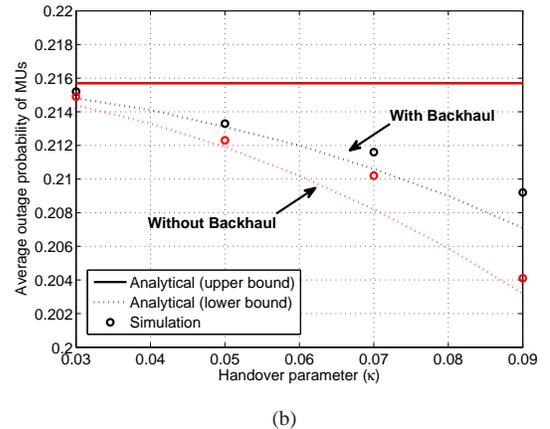}\label{vs_kappa_MBS}}
\end{center}
\begin{center}
\caption{Outage probability of MUs as
a function of the handover parameter $\kappa$ for the cases of with and without backhaul constraints a) MUs served by FAPs  b) MUs served by the MBS.} \label{vs_kappa}
\end{center}
\end{figure}

Fig.~\ref{vs_kappa} shows the average outage performance of MUs as a function of  handover parameter $\kappa$ and compares the results to the case of no  backhaul constraints.  For the case in which backhaul constraint is present,  it is assumed that $n_c = 3$. As it can be observed, in contrast to the downlink scenario \cite{CheungQ:12}, in both cases the outage probability is a monotonic function of $\kappa$. As explained in~\cite{ZeinalpourJ:14-tcom}, the  difference between the uplink and downlink arises from the fact that in the downlink scenario, as the MUs get farther away from the MBS, their received powers decrease and hence SIRs decrease as well. On the other hand,  in the uplink comunication, as they get farther away from the MBS,  due to the power control, their transmit powers  increase as well to compensate for the path loss. Naturally,  increasing the handover parameter increases the number of MUs  covered  by FAPs and hence  lowers the  co-tier interference. Note that while  the gap between the upper and lower bounds widens  as $\kappa$ increases, the  lower bound follows  the simulation results for all values of $\kappa$.

\begin{figure}[t]
\begin{center}
\includegraphics[height=50 mm]{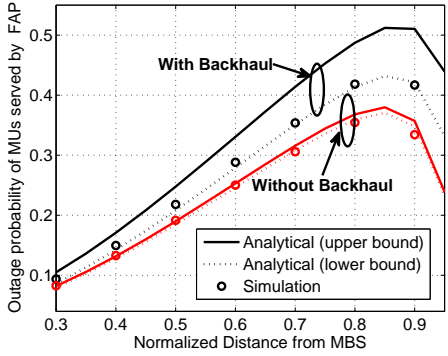}\caption{Outage probability of an MU served by an FAP as
a function of the normalized distance of the FAP from the MBS for the cases of with and without backhaul constraints.}\label{versus_df1}\end{center}
\end{figure}

Fig.~\ref{versus_df1} shows the outage probability of MUs  served by FAPs as a function of  the FAP's normalized distance from the MBS, and compares the results with the case of no backhaul restriction. Here $n_c=3$. As expected, the outage probability in the presence of backhaul  is higher that the ideal case where the FAPs have infinite backhaul capacity. The reason is that  because of the backhaul constraints fewer MUs are served by the FAPs and this leads to higher cross-tier interference. However, in both cases, at first, the outage probability  increases as the MU gets farther  from the MBS. Due to the constant received power assumption at the MBS, as the MU gets farther from the MBS, it will transmit
at a higher power, which leads to the  degradation in the performance of FUs and also MUs served by the nearby FAPs.  However, as the femtocells get close to the fringes of the cell, the outage probabilities start to  improve as well. The reason is that  femtocells that are far away from the MBS have larger coverage areas and therefore, in those regions most  MUs are serviced by nearby FAPs.

Fig.~\ref{versus_df2} shows the outage probability of MUs  served by FAPs as function of the distance between the FAP and the MBS, for different values of MUs' density ($\mu_m$). Obviously, for a fixed backhaul parameter, which is set to 3 in these curves, more MUs being served by the MBS results in higher cross-tier interference  and  hence higher  outage probabilities for MUs served by the FAPs.
\begin{figure}[t]
\begin{center}
\includegraphics[height=45 mm]{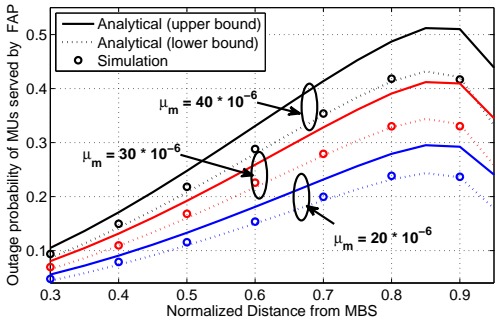}\caption{Outage probability of a MU served by a FAP as
a function of the normalized distance of the FAP from the MBS for different MUs densities.}\label{versus_df2}\end{center}
\end{figure}

\section{Conclusions}\label{sec:con}
In this paper, we have  studied two-tier cellular networks, in which each FAP has a finite backhaul capacity limiting the number of users it can serve. The MUs, FUs and FAPs have all been assumed to have  stochastic deployments according to PPPs. We have considered fixed backhaul constraints for FAPs, which limit the number of users each FAP can service. Under these assumptions,  we have derived analytical upper and lower bounds on  the outage probabilities  of MUs serviced by FAPs and MUs serviced by the MBS. All bounds have been confirmed by our simulation results.

While in our analysis we have assumed that there is only a single MBS, the results can also be applied to real networks with multiple MBSs. To do this extension, we only need to assume that each MU is assigned to its closest MBS and the macro cells  employ one of the well-known frequency reuse methods that orthogonalize neighboring cells.


\setcounter{equation}{0}
\renewcommand{\theequation}{\thesection.\arabic{equation}}

\appendices
\section{Partitioning $\Sc_m$}\label{app:A}
In this section, we briefly review the partitioning  of the coverage area presented in \cite{ZeinalpourJ:14-tcom}. Consider the MBS $b_m$ and FAP $a_f$ located at distance $d$ from each other. (Refer to Fig.~\ref{fig:kappa-quant}.) $\Sc_m$ denotes the circle of radius $R$ around $b_m$. The set of points $u$ such that $d(u,a_f)/d(u,b_m)=\kappa'$ or $d(u,a_f)/d(u,b_m)=1/\kappa'$, where $\kappa'\in(0,1)$ are two circles of radius ${\kappa'\over 1-\kappa'^2}$. In  Fig.~\ref{fig:kappa-quant}, the colored pairs of circles correspond to three different values of $\kappa'$.


\begin{figure}[b]
\begin{center}
\includegraphics[width=7cm]{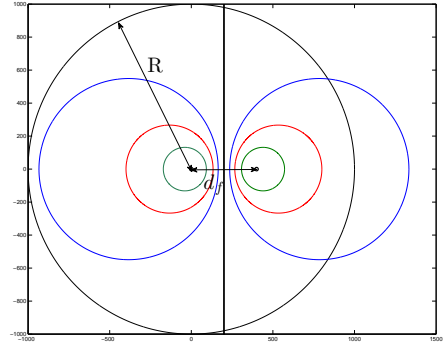}\caption{Partitioning the coverage area}\label{fig:kappa-quant}\end{center}
\end{figure}

Consider $\kappa_0,\ldots,\kappa_t$ such that $\kappa_0=\kappa<\kappa_1<\kappa_2<\ldots<\kappa_t=1$, and the $2t$ pairs of circles corresponding to  $\kappa_0,\ldots,\kappa_{t-1}$. These circles do not intersect and in addition to the line corresponding to $\kappa_t=1$, which corresponds to the set of points $u$ satisfying $d(u,a_f)=d(u,b_m)$,  partition $\Sc_m$  into $2(t+1)$ regions.
\section{Distribution of users in $\Uc_m^{\rm ns}(b_m)$}\label{sec:dist-users-nc}

As a reminder $\Uc_m^{\rm out}(b_m)$ denotes the set of users that are covered by the MBS $b_m$ because they do not fall into the coverage area of any of the FAPs. In this appendix, we prove that for $\kappa$ small the distance of the users in $\Uc_m^{\rm out}(b_m)$  to the MBS has an almost uniform distribution. In this section, we assume  that $\kappa\leq 0.5$.

Given FAP $a_f\in\Ac_f$, let $\Cc(a_f)$ denote the coverage area of $a_f$. As explained earlier, for FAP $a_f$ at distance $d$ from $b_m$, $\Cc(a_f)$ is circle of radius ${\kappa d\over (1-\kappa^2)}$, whose center is located at distance ${d\over 1-\kappa^2}$ from $b_m$ on the line connecting $b_m$ to $a_f$.

Consider user $u$ that is located uniformly at random on $\Sc_m$. Define  $\Ec$ as the event that $u$ does not fall in the coverage area of any of the FAPs, \ie 
\[
\Ec\triangleq \{u\notin \Cc(a_f), \;\forall a_f\in\Ac_f\}.
\]
Let
\[
D_u\triangleq d(u,b_m).
\]
In this section, we derive the conditional pdf of $D_u$ conditioned on  $\Ec$, $f_{D_u}(\cdot|\Ec)$. By the Bayes formula,
\begin{align}
f_{D_u}(d|\Ec)={f_{D_u}(d)\P(\Ec|D_u=d)\over \P(\Ec)}.\label{eq:cond-pdf-Du-Bayes-rule}
\end{align}
Since $u$ is drawn uniformly at random, $f_{D_u}(d)={2d\over R^2}$. On the other hand, since the FAPs are drawn according to a PPP of density $\lambda_f$, we have
\begin{align}
\P(\Ec|D_u=d)&=\sum_{n=0}^{\infty}\P(\Ec,N_{\rm fap}=n|D_u=d) \nonumber\\
&=\sum_{n=0}^{\infty}p_{N_{\rm fap}}(n) \left(\P(u\notin \Cc(a_f)|D_u=d)\right)^n\nonumber\\
&=\sum_{n=0}^{\infty}\ex^{-\bar{n}_{\rm fap}}{(\bar{n}_{\rm fap})^n\over n!}\left(\P(u\notin \Cc(a_f)|D_u=d)\right)^n\nonumber\\
&=\ex^{-\bar{n}_{\rm fap}(1-\P(u\notin \Cc(a_f)|D_u=d))}\nonumber\\
&=\ex^{-\bar{n}_{\rm fap}\P(u\in \Cc(a_f)|D_u=d)}.\label{eq:prob-e-given-Du-eq-d-1}
\end{align}

To compute $\P(u\in \Cc(a_f)|D_u=d)$ consider user $u$ at distance $d$ from $b_m$ and FAP located at distance $r$ from $b_m$.  (Refer to Fig.~\ref{fig:coverage-prob}.) In order for $u$ to be covered by $a_f$, $d$ should satisfy
\[
{r\over 1-\kappa^2}-{r\kappa\over 1-\kappa^2}\leq d\leq {r\over 1-\kappa^2}+{r\kappa\over 1-\kappa^2},
\]
or
\[
(1-\kappa)d\leq r\leq(1+\kappa)d.
\]
Given $r\in ((1-\kappa)d,(1+\kappa)d)$, the angle between the lines $(b_m,a_f)$ and $(b_m,u)$ should be within $(-\theta,\theta)$, where
\begin{align}
\cos (\theta)&= {d^2 + ({r \over 1-\kappa^2})^2- ({\kappa r \over 1-\kappa^2})^2\over {2dr \over 1-\kappa^2}}= {d^2(1-\kappa^2) + r^2\over 2dr}.\label{label:cos-thata-1}
\end{align}
Let $r=d(1+\rho)$, where $\rho\in(-\kappa,\kappa)$. Employing this change of variable, it follows from \eqref{label:cos-thata-1} that
\[
\cos (\theta)=1-{\kappa^2-\rho^2\over 2(1+\rho)},
\]
and
\begin{align*}
\sin^2 (\theta)&={\kappa^2-\rho^2\over 2(1+\rho)}(2-{\kappa^2-\rho^2\over 2(1+\rho)})\nonumber\\
&=\kappa^2-\rho^2(1-{\rho\over 2(1+\rho)}-{\kappa^2-\rho^2\over 4(1+\rho)^2}).
\end{align*}
Therefore, since for $0\leq x\leq 1$, $1-x\leq \sqrt{1-x}\leq 1$, we have
\begin{align}
\sqrt{\kappa^2-\rho^2}(1-{\rho\over 2(1+\rho)}-{\kappa^2-\rho^2\over 4(1+\rho)^2}) \leq \sin(\theta)&\leq \sqrt{\kappa^2-\rho^2}.\label{eq:bound-1}
\end{align}
But,
\begin{align}
{\rho\over 2(1+\rho)}+{\kappa^2-\rho^2\over 4(1+\rho)^2}&\leq {\kappa\over 2(1-\kappa)}+{\kappa^2\over 4(1-\kappa)^2}\leq 2\kappa,\label{eq:bound-2}
\end{align}
where the last line follows from our assumption that $\kappa\leq 0.5$. And,
\begin{align}
\int_{-\kappa}^{\kappa} 2(1+\rho)\sqrt{\kappa^2-\rho^2}d\rho=\pi \kappa^2.\label{eq:integral}
\end{align}

\begin{figure}
\begin{center}
\includegraphics[width=7.5cm]{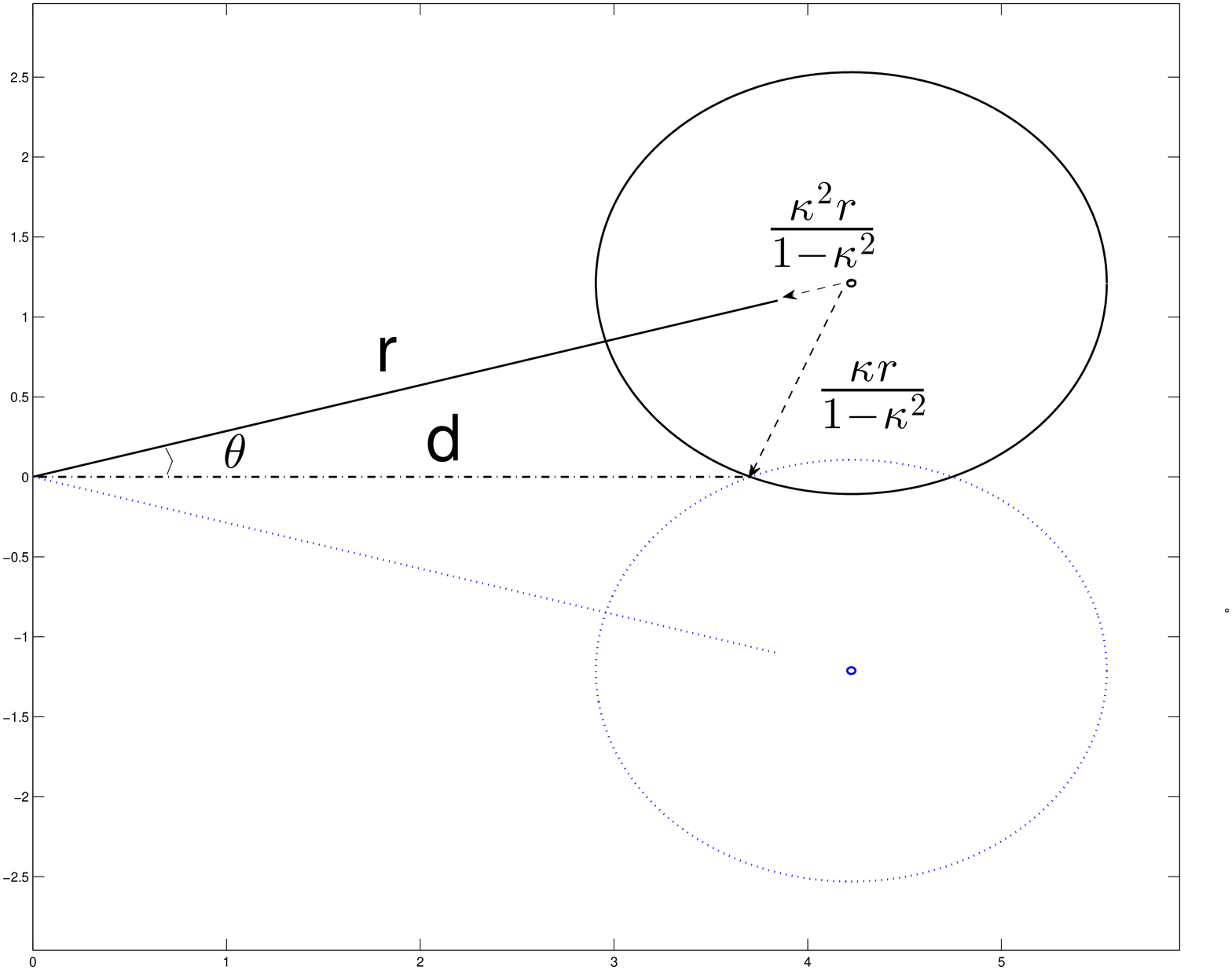}\caption{User $u$ located at distance $d$ from $b_m$ falling in the coverage area of $a_f$ at distance $r$ from $b_m$.}\label{fig:coverage-prob}
\end{center}
\end{figure}

Therefore, since $\P(u\in \Cc(a_f)|D_u=d)={d^2 \over \pi R^2} \int_{-\kappa}^{\kappa} 2(1+\rho)\sin( \theta) d\rho$, combining  \eqref{eq:bound-1}, \eqref{eq:bound-2} and \eqref{eq:integral}, it follows that
\begin{align}
(1-2\kappa){d^2\kappa^2\over R^2}\leq \P(u\in \Cc(a_f)|D_u=d)&\leq {d^2\kappa^2\over R^2}.\label{eq:prob-coverage-u-dist-d}
\end{align}
Combining \eqref{eq:prob-e-given-Du-eq-d-1} and \eqref{eq:prob-coverage-u-dist-d} yields
\begin{align}
\ex^{-\bar{n}_{\rm fap}d^2\kappa^2/R^2}\leq \P(\Ec|D_u=r)\leq \ex^{-\bar{n}_{\rm fap}d^2\kappa^2(1-2\kappa)/R^2},\label{eq:prob-e-given-Du-eq-d}
\end{align}
and $\P(\Ec)=\int_0^{R}{2r\over R^2}\P(\Ec|D_u=d)dr$ satisfies
\begin{align}
{1-\ex^{-\kappa^2 \bar{n}_{\rm fap}}\over \kappa^2 \bar{n}_{\rm fap}}\leq \P(\Ec)\leq {1-\ex^{-(1-2\kappa)\kappa^2 \bar{n}_{\rm fap}}\over (1-2\kappa)\kappa^2 \bar{n}_{\rm fap}}.\label{eq:bound-on-P-E}
\end{align}
Finally, from \eqref{eq:cond-pdf-Du-Bayes-rule}, \eqref{eq:prob-e-given-Du-eq-d} and \eqref{eq:bound-on-P-E},
\begin{align}
f_{D_u}(d|\Ec)\geq{(1-2\kappa)\kappa^2 \bar{n}_{\rm fap}\ex^{-\bar{n}_{\rm fap}d^2\kappa^2/R^2} \over 1-\ex^{-(1-2\kappa)\kappa^2 \bar{n}_{\rm fap}}}({2d\over R^2}),\label{eq:ub-lb-f-Du-cond-lb}
\end{align}
\begin{align}
f_{D_u}(d|\Ec)\leq{\kappa^2 \bar{n}_{\rm fap}\ex^{-\bar{n}_{\rm fap}d^2(1-2\kappa)\kappa^2/R^2} \over 1-\ex^{-\kappa^2 \bar{n}_{\rm fap}}}({2d\over R^2}).\label{eq:ub-lb-f-Du-cond-ub}
\end{align}
Note that for  $\kappa \ll 1$,  the lower bound and the bound bound in  \eqref{eq:ub-lb-f-Du-cond-lb} and \eqref{eq:ub-lb-f-Du-cond-ub}, respectively,  converge to $2d/R^2$, which corresponds to the uniform distribution over a circle of radius $R$.

\bibliographystyle{unsrt}
\bibliography{myrefs}

\end{document}

%% file: defns.tex
\usepackage{xspace}
\usepackage{bbm}

\newcommand{\Ac}{\mathcal{A}}
\newcommand{\Bc}{\mathcal{B}}
\newcommand{\Cc}{\mathcal{C}}

\newcommand{\Ec}{\mathcal{E}}

\newcommand{\Pc}{\mathcal{P}}

\newcommand{\Sc}{\mathcal{S}}

\newcommand{\Uc}{\mathcal{U}}

\newcommand{\uh}{{\hat{u}}}

\def\a{\alpha}

\def\e{\epsilon}

\DeclareMathOperator\E{E}
\let\P\relax
\DeclareMathOperator\P{P}

\newcommand\ie{i.e.,\xspace}
\def\textiid{i.i.d.\@\xspace}
\newcommand\iid{\ifmmode\text{ i.i.d. } \else \textiid \fi}

\newcommand{\ind}{\mathbbmss{1}}
